\documentclass[12pt,letter]{article}

 \usepackage{helvet}
\usepackage{graphicx}
\usepackage{caption}
\usepackage{amsmath}
\usepackage{natbib}
\usepackage{amsfonts}
\usepackage{amssymb}  
\usepackage{amsthm}  
\usepackage{subfigure}
\usepackage{setspace}
\usepackage{verbatim}
\usepackage{color}
\usepackage{soul}
\usepackage{url}
\usepackage{qtree}
\usepackage{algorithm}
\usepackage{algorithmic}
\usepackage{pdflscape}

\newcommand{\R}{\mathbb{R}}

\newcommand{\I}{1{\hskip -2.5 pt}\hbox{I} }

\newcommand{\Be}[2]{\mbox{Be}(#1,#2)}
\newcommand{\mockalph}[1]{}

\newtheorem{lemma}{Lemma}

\bibpunct{(}{)}{;}{a}{}{,}

\title{Multiscale Bernstein polynomials for densities}
\author{Antonio Canale\thanks{Dipartimento di Scienze Economico-Sociali e Matematico-Statistiche, Universit\`a degli studi di Torino e Collegio Carlo Alberto, Torino, Italy$\,$ \tt antonio.canale@unito.it } $\,\,$and 
David B. Dunson\thanks{Department of Statistical Science, Duke University,
Durham, NC \tt dunson@duke.edu}
}

\begin{document}

\maketitle

\abstract{
Our focus is on constructing a multiscale nonparametric prior for densities.  The Bayes density estimation literature is dominated 
by single scale methods, with the exception of Polya trees, which favor overly-spiky densities even when the truth is smooth.
We propose a multiscale Bernstein polynomial family of priors, which produce smooth realizations that do not rely on hard
partitioning of the support.  At each level in an infinitely-deep binary tree, we place a beta dictionary density; within a scale
the densities are equivalent to Bernstein polynomials.  Using a stick-breaking characterization, stochastically decreasing weights are 
allocated to the finer scale dictionary elements.  A slice sampler is used for posterior computation, and properties are described.
The method characterizes densities with locally-varying smoothness, and can produce a sequence of coarse to fine density estimates. 
An extension for Bayesian testing of group differences is introduced and applied to DNA methylation array data. 
}

{\center \textbf{Keywords: }}
Density estimation; Multiresolution; Multiscale clustering; Multiscale testing; Nonparametric Bayes; Polya tree; Stick-breaking; Wavelets

\newpage

\section{Introduction}

Multiscale estimators have well known advantages, including the ability to characterize abrupt local changes and to provide a compressed estimate to a desired level of resolution.  Such advantages have lead to enormous popularity of wavelets, which are routinely used in signal and image processing, and have had attention in the literature on density estimation.  \citet{dono:etal:1996} developed a wavelet thresholding approach for density estimation, which has minimax optimality properties, and there is a literature developing modifications for deconvolution problems \citep{pens:vida:1999}, censored data \citep{niu:2012}, time series \citep{garc:barr:2012} and other settings.  \citet{lock:pete:2013} proposed an approach, which can better characterize local symmetry and other features commonly observed in practice, using multiwavelets.  \citet{chen:etal:2012} instead use geometric multiresolution analysis methods related to wavelets to obtain estimates of high-dimensional distributions having low-dimensional support.

Although there is a rich Bayesian literature on multiscale function estimation \citep{abra:etal:1998, clyd:etal:1998, clyd:geor:2000, wang:etal:2007}, there has been limited consideration of  Bayesian multiscale density estimation.  Popular methods for Bayes density estimation rely on kernel mixtures.  For example, Dirichlet process mixtures are applied routinely.  By using location-scale mixtures, one can accommodate varying smoothness, with the density being flat in certain regions and concentrated in others.  However, Dirichlet processes lack the appealing multiscale structure.  Polya trees provide a multiscale alternative \citep{maul:etal:1992,lavi:1992a,lavi:1992b}, but have practical disadvantages.  They tend to produce highly spiky density estimates even when the true density is smooth, and have sensitivity to a pre-specified partition sequence.  This sensitivity can be ameliorated by mixing Polya trees \citep{hans:2002}, but at the expense of more difficult computation.  

Our focus is on developing a new approach for Bayesian multiscale density estimation, which inherits many of the advantages of Dirichlet process mixtures while avoiding the key disadvantages of Polya trees.  We want a framework that is easily computable, has desirable multiscale approximation properties, allows centering on an initial guess at the density, and can be extended in a straightforward manner to include covariates and allow embedding within larger models. We accomplish this using a multiscale extension of mixtures of Bernstein polynomials \citep{petr:1999a, petr:1999b}, which have been shown to have appealing asymptotic properties in the single scale case \citep{petr:wass:2002,ghos:2001}.

In the next section, our multiscale prior for densities is introduced and properties are discussed. Section 3 introduces posterior computation via a slice sampling algorithm. In Section 4 the performance of the method in terms of density estimation is evaluated via a simulation study. Section 5 discusses generalizations, with particular emphasis on Bayesian multiscale inferences on differences between groups. Section 6 applies the method to a DNA methylation array dataset on breast cancer, and Section 7 concludes. Proofs and computational details are reported in the Appendix.  

\section{Multiscale priors for densities}

\subsection{Proposed model}
\label{sec:model}

Let $x \in \mathcal{X} \subset \R$ be a random variable having density $g$ with respect to Lebesgue measure.  Assume that $g_0$ is a prior guess for $g$, with $G_0$ and $G_0^{-1}$ the corresponding cumulative distribution function (CDF) and inverse CDF, respectively.  We induce a prior $g \sim \Pi$ centered on $g_0$ through a prior for the density $f$ of $y = G_0( x ) \in (0,1)$.  The CDFs $F$ and $G$ corresponding to the densities $f$ and $g$, respectively, have the following relationship
\begin{eqnarray}
G(x) = F\{ G_0(x) \}, x \in \mathcal{X},\quad 
F(y) = G\{ G_0^{-1}(y) \}, y \in (0,1). \label{eq:map}
\end{eqnarray}
We assume that $f$ follows a multiscale mixture of Bernstein polynomials, 
\begin{equation}
	f(y) = \sum_{s=0}^\infty \sum_{h=1}^{2^s} \pi_{s,h} \mbox{Be}(y; h, 2^s - h +1),
\label{eq:mix1}
\end{equation}
where Be($a$, $b$) denotes the beta density with mean $a/(a+b)$, and $\{ \pi_{s,h} \}$ are random weights drawn from a suitable stochastic process.
We introduce an infinite sequence of scales $s=0,1,\ldots,\infty$.  At scale $s$, we include $2^s$ Bernstein polynomial basis densities.  The framework can be represented as a binary tree in which  each layer is indexed by a scale and each node is a suitable beta density. 
For example, at the root node, we have the Be(1,1) density which generates two daughters Be(1,2) and Be(2,1) and so on. In general, let $s$ denote the scale and $h$ the polynomial within the scale. The node $(s,h)$ in the tree is related to the Be($h, 2^s - h +1$) density. A cartoon of the binary tree is reported in Figure~\ref{tree1}.

\begin{figure}
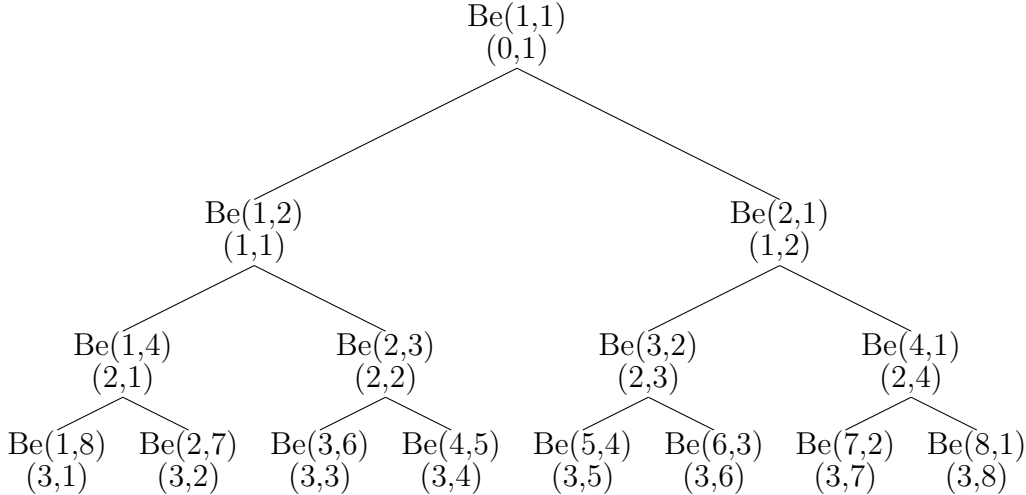

\label{tree1}
\Tree [.Be(1,1)\\(0,1)
		[.Be(1,2)\\(1,1) 
			[.Be(1,4)\\(2,1) 
				[.Be(1,8)\\(3,1)  ] [.Be(2,7)\\(3,2)  ] ] 
			[.Be(2,3)\\(2,2) 
				[.Be(3,6)\\(3,3)  ] [.Be(4,5)\\(3,4)  ] ] ] 
		[.Be(2,1)\\(1,2) 
			[.Be(3,2)\\(2,3) 
				[.Be(5,4)\\(3,5)  ] [.Be(6,3)\\(3,6)  ]  ] 
			[.Be(4,1)\\(2,4) 
				[.Be(7,2)\\(3,7)  ] [.Be(8,1)\\(3,8)  ]  ] ]    
	]
\caption{Binary tree with beta kernels at each node $(s, h)$, where $s$ is the scale level and $h$ is the index within the scale}
\end{figure}

A prior measure for the multiscale mixture \eqref{eq:mix1} is obtained by specifying a stochastic process for the infinite dimensional set of weights $\{\pi_{s,h}\}$. To this end we introduce, for each scale $s$ and node $h$ within the scale, independent random variables
\begin{equation}
S_{s,h} \sim \Be{1}{a},\quad  R_{s,h} \sim \Be{b}{b},
\label{eq:SR}
\end{equation}
corresponding to the probability of stopping and taking the right path conditionally on not stopping, respectively. Define the weights as 
\begin{equation}
	\pi_{s,h} = S_{s,h} \prod_{r<s} (1-S_{r,g_{shr}}) T_{shr}
\label{eq:weights}
\end{equation}
where $g_{shr} = \lceil h/2^{s-r} \rceil$ is the node traveled through at scale $r$ on the way to node $h$ at scale $s$, $T_{shr} = R_{r,g_{shr}}$ if $(r+1,g_{shr+1})$ is the right daughter of node $(r,g_{shr})$, and $T_{shr} = 1-R_{r,g_{shr}}$ if $(r+1,g_{shr+1})$ is the left daughter of $(r,g_{shr})$. For binary trees, there is a unique path leading from the root node to node $(s,h)$, and $\cal T$ denotes the infinite deep binary tree of the weights \eqref{eq:weights}.  We refer to the prior  resulting from \eqref{eq:mix1}--\eqref{eq:weights} as a multiscale Bernstein polynomial (msBP) prior and we write $f \sim \mbox{msBP}(a,b)$.  The choice for the hyperparameters are discussed in the next section.

The infinite tree of probability weights is generated from a generalization of the stick-breaking process representation of the Dirichlet process \citep{art:seth:1994}.  Each time the stick is broken, it is consequently  randomly divided in two parts (one for the probability of going right, the remainder for the probability of going left) before the next break. 
An alternative treed stick-breaking process is proposed by \citet{adam:etal:2010} where a first stick-breaking process defines the vertical growth of an infinitely wide tree and a second puts weights on the infinite number of descendant nodes.

Sampling a random variable $y$ from a random density, which is generated from a msBP prior, can be described as follows. At node $(s,h)$, generate a random probability $S_{s,h} \sim \Be{1}{a}$ corresponding to the probability of stopping at that node given you passed through that node, and $R_{s,h} \sim \Be{b}{b}$  corresponding to the probability of taking the right path in the tree in moving to the next finer scale given you did not stop at node $(s,h)$. Conditionally on being at the node $(s,h)$ we assume that $y \sim \mbox{Be}(y; h, 2^s-h+1)$. Algorithm~1 describes how to generate $y$ from an msBP density.

\begin{algorithm}
\caption{Generating a draw from a random density having an msBP prior}
\begin{algorithmic}
\footnotesize
\STATE \texttt{loop} = TRUE;
\STATE $s = 0$, $h=1$;
\WHILE{\texttt{loop}}
\STATE let \texttt{loop} = FALSE with probability $S_{s,h}$.
\IF{\texttt{loop}}
	\STATE with probability $R_{s,h}$, let $h = 2h$
	\STATE with probability $1-R_{s,h}$, let $h = 2h-1$
\ENDIF
\ENDWHILE
\STATE generate $y \sim \mbox{Be}(h, 2^{s} - h +1)$.
\end{algorithmic}
\end{algorithm}

\subsection{Basic properties} 
\label{sec:properties}

In this section we study basic properties of the proposed prior. A first requirement is that the construction  leads to a meaningful sequence of weights. The next lemma shows that the random weights on each node of the infinitely deep tree sum to one almost surely.

\begin{lemma}
\label{lem:sumtoone}
Let $\pi_{s,h}$ be an infinite sequence of weights defined as in \eqref{eq:SR}--\eqref{eq:weights}. Then,
\begin{equation}
	 \sum_{s=0}^\infty \sum_{h=1}^{2^s} \pi_{s,h} = 1 
\label{eq:sumtoone}
\end{equation}
almost surely for any $a,b>0$.
\end{lemma}

The total weight placed on a scale $s$ is controlled by the prior for $S_{s,h}$. The expected probability allocated to node $h$ at scale $s$ can be expressed as 
\begin{eqnarray}
\mbox{E}(\pi_{s,h}) & = & \mbox{E}\bigg\{ S_s \prod_{l=0}^{s-1} (1-S_l) \prod_{l=1}^s T_l \bigg\} \nonumber \\
& = & \bigg( \frac{1}{1 + a} \bigg) \bigg( \frac{a}{1+a} \bigg)^s \bigg( \frac{1}{2} \bigg)^s = \frac{1}{1+a} \bigg( \frac{a}{2 + 2a} \bigg)^s, 
\label{eq:mean}
\end{eqnarray}
where we discard the $h$ subscript on $S_l \sim \mbox{Be}(1,a)$ and $T_l \sim \mbox{Be}(b,b)$ for ease in notation.  This does not impact the calculation because any path taken up to scale $s$ has the same probability {\em a priori} and the random variables in \eqref{eq:SR} have the same distribution regardless of the path that is taken.  Similarly 
\[
\mbox{E}(\pi_{s,h}^2)  =  \mbox{E}\bigg\{ S_s^2 \prod_{l=0}^{s-1} (1-S_l)^2 \prod_{l=1}^s T_l^2 \bigg\}, 
 =  \frac{2}{(1 + a)(2+a)} \bigg( \frac{a}{2+a} \bigg)^s \bigg\{ \frac{b+1}{2(2b+1)} \bigg\}^s.
\]
Hence at scale $s=0$ the variance is 
$\mbox{Var}(\pi_{0,1}) =  a/\{(2+a)(1 + a)^2\}$, 
while for $s>0$
\begin{eqnarray}
\mbox{Var}(\pi_{s,h}) & = &  \frac{2}{(1 + a)(2+a)} \bigg( \frac{a}{2+a} \bigg)^s \bigg\{ \frac{b+1}{2(2b+1)} \bigg\}^s - 
\left\{\frac{1}{1+a} \bigg( \frac{a}{2 + 2a} \bigg)^s \right\}^2.
\label{eq:variance}
\end{eqnarray}

We can additionally verify that our prior for the CDF $G$ is centered on the chosen $G_0$.  Letting $F(A) = \int_A f$, we obtain $E\{F(A)\} = \lambda(A)$, where $\lambda(A)$ is the Lebesgue measure over the set $A$. Details are reported in the Appendix.  Hence, the prior for the density of $y$ is automatically centered on a uniform density on $[0,1]$. 
This is the desired behavior as $y \sim \mbox{Unif}(0,1)$ with $x = G_0^{-1}( y )$ implies that $x \sim g_0$, which is our prior guess for the observed data density.  In addition, from (\ref{eq:map}), $\mbox{E}\{ F(y) \} = y$ implies  
\[
\mbox{E}[ G\{ G_0^{-1}(y) \}] = y = \mbox{E}\{ G( x ) \} = G_0(x),
\]
so that the prior expectation for the CDF $G$ is $G_0$ as desired.
	
From equation \eqref{eq:mean} and \eqref{eq:variance}, the hyperparameter $a$ controls the decline in probabilities over scales. In general, letting $S^{(i)}$ denote the scale at which the $i$th observation falls, we have 
\[
	E(S^{(i)}) = \sum_{s=0}^\infty s \frac{1}{1+a} \bigg( \frac{a}{2 + 2a} \bigg)^s = a.
\]
Hence, the value of $a$ is the expected scale from which observations are drawn. For small $a$, high probability is placed on coarse scales, leading to smoother densities, with $a \to 0$ inducing $\pi_{0,1}=1$ and hence $f(y)$ uniform.  As $a$ increases, finer scale densities will be weighted higher, leading to spiker realizations.  To illustrate this, Figure~\ref{fig:realizations1} shows realizations from the prior for different $a$ values. To better isolate the contribution of the $a$ hyperparameter, we fixed the realizations of $R_{s,h} \sim \Be{1}{1}$ for all subplots. 

\begin{figure}
\begin{center}
\includegraphics[scale=.8]{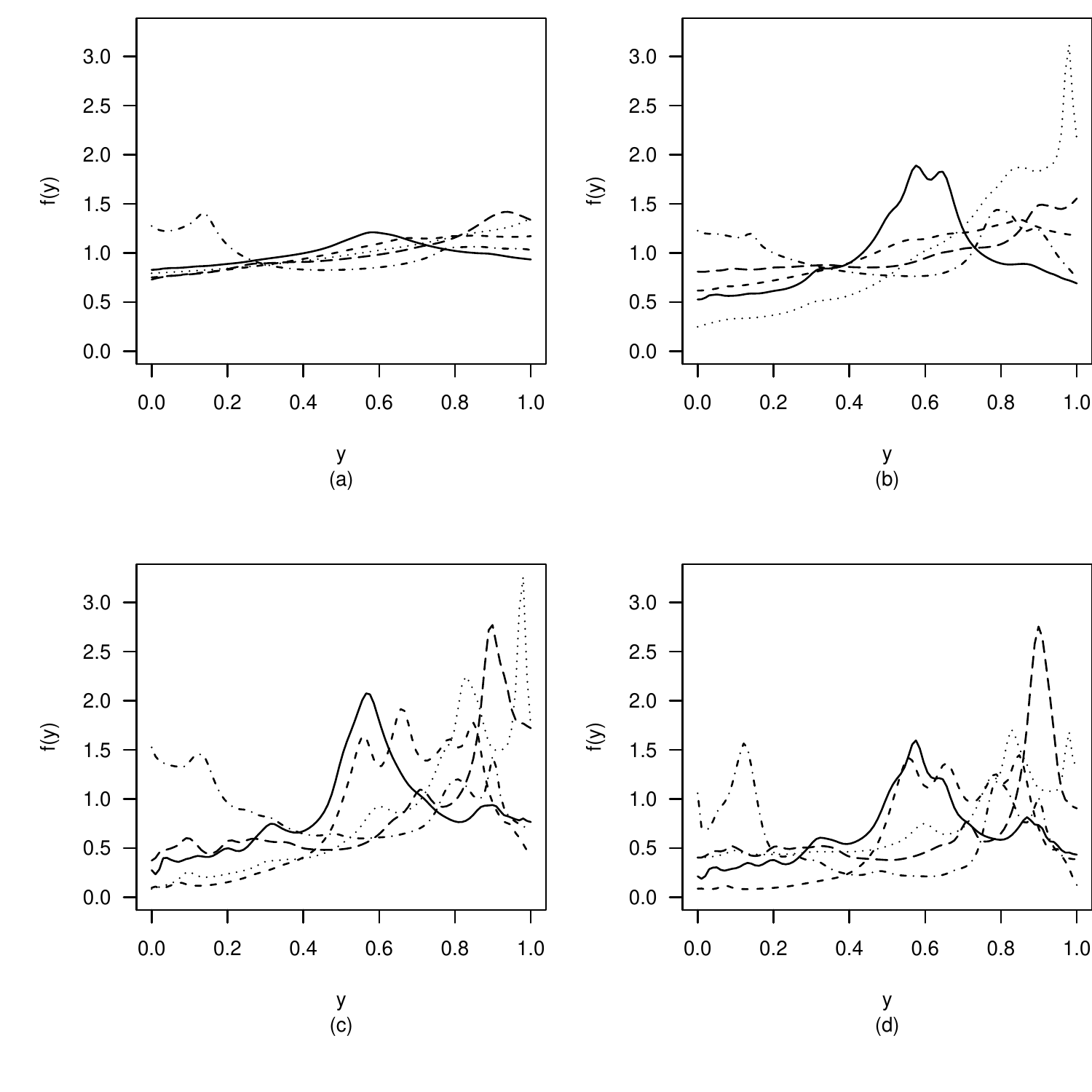}
\end{center}
\caption{Five realizations from an msBP prior with $b=1$ and  (a) $a=1$, (b) $a=2$, (c) $a=5$, and (d) $a=10$.}
\label{fig:realizations1}
\end{figure}

An appealing aspect of the proposed formulation is that individuals sampled from a distribution that is assigned an msBP prior are allocated to clusters in a multiscale fashion.  In particular, two individuals having similar observations may have the same cluster allocation up to some scale $s$, but perhaps are not clustered on finer scales.  Clustering is intrinsically a scale dependent notion, and our model is the first to our knowledge to formalize multiscale clustering in a model based probabilistic manner.  Under the above structure, the probability that two individuals $i$ and $i'$ are assigned to the same scale $s$ cluster is one for $s=0$ and for $s>0$, is equal to  
\begin{eqnarray}
\lefteqn{ 2^s \mbox{E}\bigg\{ \prod_{l=0}^{s-1} (1-S_l)^2 T_l^2 \bigg\} } \nonumber \\ 
& & = 2^s \prod_{l=0}^{s-1} \mbox{E}( \overline{S}_l^2 ) \mbox{E}( T_l^2 )  = 2^s \bigg( \frac{a}{ a+ 2} \bigg)^s \bigg( \frac{1}{2} \bigg)^s  \bigg( \frac{b+1}{2b + 1} \bigg)^s 
= \bigg\{ \bigg( \frac{a}{a+2} \bigg) \bigg( \frac{b+1}{2b+1}\bigg) \bigg\}^s. \nonumber
\end{eqnarray}
This is derived by calculating the expected probability that two individuals travel though node $h$ at scale $s$ and multiplying by the number of nodes in scale $s$.  This form is intuitive.  As $b \to 0$, the Be($b,b$) density degenerates to $0.5\delta_0 + 0.5\delta_1$, so that variability among subjects in the chosen paths through the tree decreases and all subjects take a common path chosen completely at random via unbiased coin flips at each node.  In such a limiting case, $(b+1)/(2b+1) \to 1$ and the probability of clustering subjects at scale $s$ is simply the probability of surviving to that scale and not being allocated to a coarser scale component.  At the other extreme, as $b \to \infty$ each subject independently flips an unbiased coin in deciding to go right or left at each node of the tree, and $(b+1)/(2b+1) \to 1/2$.  Hyperpriors can be chosen for $a$ and $b$ to allow the data to inform about these tuning parameters; we find that choosing a hyperprior for $a$ is particularly important, with
$b=1$ as a default.

Approximations of the msBP process can be obtained fixing an upper bound $s$ for the depth of the tree. The truncation is applied by pruning $\mathcal T$ at scale $s$, setting $S_{s,h} = 1$ for each $h = 1, \dots, 2^s$ as done in \citet{art:ishw:jame:2001} and related works in the single scale case.
We denote the scale $s$ approximation as  
\begin{eqnarray}
f^s(y) = \sum_{l=0}^s \sum_{h=1}^{2^l} \tilde{\pi}_{l,h} \mbox{Be}( y; h, 2^l - h + 1), \label{eq:fs}
\end{eqnarray}
with $\tilde{\pi}_{l,h}$ identical to $\pi_{l,h}$ except that we set all the stopping probabilities at scale $s$ equal to one to ensure that the weights sum to one and that $f^s(y)$ is a valid probability density on $\mathcal{Y} = [0,1]$.  Let $\mathcal{T}^s$ denote the pruned binary tree of weights. It is interesting to study the accuracy of the approximation of $f^s(y)$ to $f(y)$ as the scale $s$ changes under different metrics.  For example, using the total variation distance, \begin{eqnarray}
\lefteqn{ d_{TV}(P_s, P) = \sup_{B \in \mathcal{B}} | P^s(B) - P(B) | } \nonumber \\
& & = \sup_{B \in \mathcal{B}} \bigg| \sum_{h=1}^{2^s} \tilde{\pi}_{s,h} \mbox{Be}( B; h, 2^s - h + 1) - 
\sum_{l=s}^{\infty}\sum_{h=1}^{2^l} \pi_{l,h} \mbox{Be}( B; h, 2^l-h+1) \bigg|,
\label{eq:dTV}
\end{eqnarray}
where $P^s( B ) = \int_B f^s(y)dy$ and $P(B) = \int_B f(y)dy$, for all $B \in \mathcal{B}$, denote the probability measures corresponding to densities $f^s(y)$ and $f(y)$, respectively, with $\mathcal{B}$ the Borel $\sigma$-algebra of subsets of $\mathcal{Y} = [0,1]$. The next lemma shows that {\em a priori} the expected deviation of the truncation approximation $P^s$ from $P$ is zero and the variance is decreasing exponentially with $s$. 
\begin{lemma}
\label{lem:tvd}
The expectation of the total variation distance between $P^s(B)$ and $P(B)$ is zero and its variance is 
\[
	\mbox{Var}\left\{d_{TV}(P_s, P)\right\} = 2\left(\frac{a}{a+1}\right)^s.
\]
\end{lemma}

\section{Posterior computation}

In this section we demonstrate that a straightforward Markov chain Monte Carlo (MCMC)
algorithm can be constructed to perform posterior inference under the msBP prior. 
The algorithm consists of two primary steps: (i) allocate each observation 
to a multiscale cluster, conditionally on the current values of the probabilities $\{ \pi_{s,h} \}$; 
(ii) conditionally on the cluster allocations, update the probabilities.

Suppose subject $i$ is assigned to node $(s_i,h_i)$, with $s_i$ the scale and $h_i$ the node within scale.
Conditionally on $\{ \pi_{s,h} \}$, the posterior probability of subject $i$ belonging to node ($s,h$) is simply
\begin{align*}
\mbox{pr}(s_i = s, h_i=h | y_i, \pi_{s,h}) & \propto \pi_{s,h} \mbox{Be}(y; h, 2^s - h +1).
\end{align*}
Consider the total mass assigned at scale $s$, defined as $\pi_s = \sum_{h=1}^{2^s} \pi_{s,h}$, and let $\bar{\pi}_{s,h} = \pi_{s,h}/\pi_{s}$. Under this notation, we can rewrite \eqref{eq:mix1} as
\begin{align*}
	f(y) = \sum_{s=0}^\infty \pi_{s} \sum_{h=1}^{2^s} \bar{\pi}_{s,h} \mbox{Be}(y; h, 2^s - h +1).
\end{align*}
To allocate each subject to a multiscale cluster, we rely on a multiscale modification of the slice sampler of \citet{kall:etal:2011}. Consider the joint density 
\[
	f(y_i,u_i,s_i) \propto \I(u_i<\pi_{s_i}) \sum_{h=1}^{2^{s_i}} \bar{\pi}_{s_i,h}  \mbox{Be}(y_i; h, 2^{s_i} - h +1 ).
\]
The full conditional posterior distributions are
\begin{align}
	& u_i | y_i, s_i \sim U(0, \pi_{s_i}), \label{eq:step1slice} \\
	& \mbox{pr}(s_i=s | u_i,y_i) \propto \I(s: \pi_{s}> u_i)\sum_{h=1}^{2^{s}} \bar{\pi}_{s,h}  \mbox{Be}(y_i; h, 2^{s} - h +1 ),\label{eq:step1bmultinom}\\
	& \mbox{pr}(h_i=h | u_i,y_i,s_i) \propto \bar{\pi}_{s_i,h}  \mbox{Be}(y_i; h, 2^{s_i} - h +1 ).
	\label{eq:step1bmultinom2}
\end{align}

Even with an infinite resolution level, equation \eqref{eq:step1bmultinom} implies that observations are assigned to a finite number of scales and there are a finite number of probabilities to evaluate.  Conditionally on the scale, equation \eqref{eq:step1bmultinom2} induces a simple multinomial sampling, which allocates a subject to a particular node within that scale. Algorithm 2 summarizes the posterior cluster allocation step.
An alternative version of this slice sampler considers the joint density 
\[
	f(y_i,u_i,s_i,h_i) \propto \I(u_i<\pi_{s_i,h_i}) \mbox{Be}(y_i; h_i, 2^{s_i} - h_i +1 ), 
\]
leading to conditional posteriors
\[
	 u_i | y_i, s_i, h_i \sim U(0, \pi_{s_i,h_i}), \quad
	\mbox{pr}(s_i=s,h_i=h |u_i,y_i) \propto \I(\pi_{s,h}> u_i)\mbox{Be}(y_i; h, 2^{s} - h +1 ).
\]
In the second version a greater number of probabilities need to be evaluated for each subject. Our experience suggests that the sampler obtained using \eqref{eq:step1slice}--\eqref{eq:step1bmultinom2}, summarized in Algorithm~2, is more efficient and converges faster.

\begin{algorithm}
\caption{Multiscale cluster posterior allocation for $i$th subject
}
\begin{algorithmic}
\footnotesize
\FOR{each scale $s$}
	\STATE calculate $\pi_s = \sum_{h=1}^{2^s} \pi_{s,h}$:
\ENDFOR
\STATE simulate $u_i | y_i, s_i \sim U(0, \pi_{s_i})$;
\FOR{each scale $s$}
	\IF{$\pi_{s}> u_i$}
		\FOR{$h=1, \dots 2^s$}
			\STATE  compute $\bar{\pi}_{s,h} = \pi_{s,h}/\pi_{s}$
		\ENDFOR
		\STATE compute $ \mbox{pr}(s_i=s | u_i,y_i) \propto \sum_{h=1}^{2^s} \bar{\pi}_{s,h}  \mbox{Be}(y_i; h, 2^{s} - h +1 )$
	\ELSE 
		\STATE $\mbox{pr}(s_i=s | u_i,y_i) = 0$;
	\ENDIF
\ENDFOR
\STATE sample $s_i$ with probability $\mbox{pr}(s_i=s | u_i,y_i)$;
\STATE sample $h_i$ with probability $\mbox{pr}(h_i=h | y_i,s_i) \propto \bar{\pi}_{s_i,h}  \mbox{Be}(y_i; h, 2^{s_i} - h +1 )$;
\end{algorithmic}
\label{algo:postcluster}
\end{algorithm}

Conditionally on cluster allocations, we sample all the stopping and descending-right probabilities from their full conditional posterior distributions: 
\begin{equation}
S_{s,h} \sim \Be{1+n_{s,h}}{a + v_{s,h} - n_{s,h}},\quad
R_{s,h}  \sim \Be{b+r_{s,h}}{b + v_{s,h}  - n_{s,h} - r_{s,h} }, \label{eq:postSR}
\end{equation}
where $v_{s,h}$ is the number of subjects passing through node $(s,h)$, $n_{s,h}$ is the number of subjects stopping at node $(s,h)$,   and  $r_{s,h}$ is the number of subjects that continue to the right after passing through node $(s, h)$. Calculation of $v_{s,h}$ and $r_{s,h}$ can be performed via parallel computing due to the binary tree structure, improving efficiency.

If hyperpriors for $a$ and $b$ are assumed, additional sampling steps are required. Assuming $a \sim \mbox{Ga}(\beta,\gamma)$, its full conditional posterior is
\begin{equation}
	 a | -  \sim \mbox{Ga}\left(\beta + 2^{s'+1} - 1, \gamma - \sum_{s=0}^{s'} \sum_{h=1}^{2^s} \log(1-S_{s,h}) \right), 
\label{eq:posterio_a}
\end{equation}
while if $b \sim \mbox{Ga}(\delta, \lambda)$ its full conditional posterior is proportional to 
\begin{equation}
	b^\delta \prod_{s=0}^{s'} \prod_{h=1}^{2^s} \frac{1}{B(b,b)} \exp \left\{-b \left(
	\lambda \sum_{s=0}^{s'} \sum_{h=1}^{2^s} \log\{R_{s,h} (1 - R_{s,h} )\} \right) \right\},
\label{eq:posterio_b}
\end{equation}
where $s'$ is the maximum occupied scale and $B(p, q)$ is the Beta function. To sample
from the latter distribution, a Metropolis-Hastings step 
 is required. The Gibbs sampler iterates the steps outlined in Algorithm ~\ref{algo:gibbs}.

\begin{algorithm}
\caption{Gibbs sampler steps for posterior computation under msBP prior}
\begin{algorithmic}
\footnotesize
\FOR{ $i = 1, \dots, n$}
\STATE assign observation $i$ to a cluster $(s_i, h_i)$ as in Algorithm~\ref{algo:postcluster}.
\ENDFOR
\STATE compute $n_{s,h}$ the number of subjects in cluster $(h,s)$ for all occupied clusters;
\STATE compute $v_{s,h}$ the number of subjects that pass through node $(h,s)$;
\STATE compute $r_{s,h}$ the number of subjects that proceed down to the right at node $(h,s)$;
\STATE let $s_{\text{MAX}}$ be the maximum occupied scale;
\FOR{ $s = 0, \dots, s_{\text{MAX}}$}
\FOR{ $h = 1, \dots, 2^s$}
\STATE update $S_{s,h} \sim \Be{1+n_{s,h}}{a + v_{s,h} - n_{s,h}}$
\STATE update $R_{s,h} \sim \Be{b+r_{s,h}}{b + v_{s,h} - n_{s,h} - r_{s,h} }$
\ENDFOR
\ENDFOR
\STATE update $a$ from \eqref{eq:posterio_a};
\STATE update $b$ from \eqref{eq:posterio_b}.
\end{algorithmic}
\label{algo:gibbs}
\end{algorithm}

\section{Simulation study} 
\label{sec:simulation}

We compared our msBP method to standard Bayesian nonparametric techniques including DP location-scale mixtures of Gaussians, DP mixtures of Bernstein polynomials, and mixtures of Polya trees, all using the R package \texttt{DPpackage}.  In addition, we implemented a frequentist wavelet density estimator using the package \texttt{WaveThresh}, and a simple frequentist kernel estimator. Several simulations have been run under different simulation settings leading to qualitatively similar results.
We report the results for four scenarios. Scenario 1 simulated data from a mixture of betas, 0.6Be(3, 3) + 0.4Be(21, 5); Scenario 2 used a mixture of Gaussians, $0.5N (0, 4) + 0.3N (2, 1) + 0.2N (1.5, 0.25)$; Scenario 3 generated data from a density supported on the positive real line, a mixture of a gamma and a left truncated normal, $0.9\mbox{Ga}(2, 2) + 0.1N_{\mbox{\tiny LT}} (4, 0.4)$; finally, Scenario 4 generated data from a symmetric density with two spiky modes,  $0.7N (0, 4) + 0.1N (0.5, 0.01) + 0.2N (1.5, 0.4)$. 

For each case, we generated sample sizes of $n = 25, 50, 100$. Each of the approaches were applied to $200$ replicated data sets under each scenario. The methods were compared based on a Monte Carlo approximation to the mean Kolmogorov-Smirnov distance (KS), $L_1$ and $L_2$ distances.


To implement Algorithm 3, we exploit the binary tree structure of our modelling framework using efficient {\small C}\texttt{++} code embedded into R functions. In implementing the Gibbs sampler, the first 1{,}000 iterations were discarded as a burn-in
and the next 2{,}000 samples were used to calculate the posterior mean of the density on a fine grid of points. To center our prior, using a default empirical Bayes approach, we set $g_0$ equal to a kernel estimate. For the hyperparameters we fixed $b = 1$ and let $a \sim \mbox{Ga}(5, 0.5)$. We truncated the depth of the binary tree to the sixth scale. The values of the density for a wide variety of points in the domain were monitored to gauge rates of apparent convergence and mixing. The trace plots showed excellent mixing, and the \citet{gewe:1992} diagnostic suggested rapid convergence.

\begin{landscape}
\begin{table}
\begin{center}
\caption{Mean Kolmogorov Smirnoff (KS) distance, mean $L_1$ distance ($L_1$), and mean $L_2$ distance ($L_2$) between the true densities and the posterior msBP estimate (msBP), posterior DP mixture of Gaussians estimate (DPM), posterior DP mixture of Bernstein Polynomials estimate (DPB), posterior Polya's Tree estimate (PT), frequentist wavelet estimate (W), and frequantist kernel smoothing estimate (K) for Scenario 1 (S1), Scenario 2 (S2), Scenario 3 (S3), and Scenario 4 (S4). Mean distances computed over 200 samples, with Monte Carlo error in parenthesis} \scriptsize											
\begin{tabular}{ll|rrr|rrr|rrr} \hline											
 &	& 	\multicolumn{3}{c}{$n=25$} &			\multicolumn{3}{c}{$n=50$} &			\multicolumn{3}{c}{$n=100$} \\ 			
&	& 	\multicolumn{1}{c}{KS} & 	\multicolumn{1}{c}{$L_1$} & 	\multicolumn{1}{c}{$L_2$} & 	\multicolumn{1}{c}{KS} & 	\multicolumn{1}{c}{$L_1$} & 	\multicolumn{1}{c}{$L_2$} & 	\multicolumn{1}{c}{KS} & 	\multicolumn{1}{c}{$L_1$} & 	\multicolumn{1}{c}{$L_2$} \\ \hline	
S1 &	msBP &	0.9616 (0.28) &	15.3337 (3.79) &	9.0286 (4.06) &	0.8529 (0.20) &	12.4909 (2.78) &	5.8835 (2.59) &	0.7318 (0.20) &	10.2247 (2.44) &	4.1602 (1.96) \\	
&	DPM &	1.5785 (0.16) &	18.1684 (1.78) &	15.659 (2.76) &	1.4137 (0.15) &	18.1139 (1.50) &	13.4228 (2.39) &	1.3558 (0.17) &	18.2278 (1.47) &	13.0673 (2.46) \\	
&	DPBP &	1.2443 (0.19) &	22.6341 (2.42) &	15.9829 (3.53) &	0.9245 (0.27) &	15.3053 (3.83) &	8.2186 (4.03) &	0.6147 (0.24) &	9.7378 (3.12) &	3.3916 (2.16) \\	
&	PT & 	2.4917 (0.00) &	952.1645 (2.14) &	1391.3997 (2.74) &	2.4917 (0.01) &	951.1084 (1.26) &	1389.2295 (1.43) &	2.4917 (0.01) &	951.5270 (0.94) &	1389.7410 (1.11) \\	
&	W &	1.6867 (0.05) &	26.5373 (0.81) &	23.2277 (1.27) &	1.6481 (0.04) &	25.8640 (0.73) &	22.1622 (1.15) &	1.6425 (0.03) &	25.7625 (0.54) &	21.9891 (0.84) \\	
&	K & 	1.0629 (0.24) &	15.8933 (3.44) &	9.4448 (3.47) &	0.8812 (0.21) &	12.6056 (2.78) &	6.0769 (2.70) &	0.7623 (0.19) &	10.3960 (2.53) &	4.3419 (1.99) \\	\hline
S2 &	msBP &	0.0947 (0.03) &	1.5028 (0.32) &	0.0812 (0.03) &	0.0742 (0.02) &	1.1060 (0.26) &	0.0441 (0.02) &	0.0642 (0.01) &	0.9616 (0.18) &	0.0327 (0.01) \\	
&	DPM &	0.1385 (0.06) &	1.7884 (0.53) &	0.1389 (0.08) &	0.1012 (0.04) &	1.3192 (0.40) &	0.0728 (0.04) &	0.0700 (0.03) &	0.9485 (0.30) &	0.0372 (0.02) \\	
&	DPBP &	0.2339 (0.01) &	4.3513 (0.05) &	0.6461 (0.01) &	0.2339 (0.01) &	4.4880 (0.07) &	0.6648 (0.01) &	0.2339 (0.01) &	4.5672 (0.07) &	0.6783 (0.01) \\	
&	PT & 	0.2347 (0.01) &	94.2408 (0.44) &	13.9915 (0.06) &	0.2339 (0.01) &	93.9891 (0.33) &	13.9568 (0.03) &	0.2339 (0.01) &	93.8067 (0.28) &	13.9393 (0.02) \\	
&	W &	0.1424 (0.05) &	2.1501 (0.66) &	0.1756 (0.10) &	0.1027 (0.03) &	1.5620 (0.44) &	0.0917 (0.05) &	0.0717 (0.02) &	1.1410 (0.31) &	0.0468 (0.02) \\	
&	K & 	0.0931 (0.02) &	1.4714 (0.31) &	0.0767 (0.03) &	0.0778 (0.02) &	1.1730 (0.26) &	0.0485 (0.02) &	0.0665 (0.02) &	0.9893 (0.18) &	0.0344 (0.01) \\	\hline
S3 &	msBP &	0.2806 (0.05) &	2.7758 (0.77) &	0.3854 (0.18) &	0.2571 (0.04) &	2.2984 (0.64) &	0.2770 (0.13) &	0.2252 (0.03) &	1.8722 (0.43) &	0.1907 (0.07) \\	
&	DPM &	0.2494 (0.07) &	2.8651 (0.70) &	0.3922 (0.20) &	0.2276 (0.06) &	2.3452 (0.58) &	0.2760 (0.14) &	0.1938 (0.05) &	1.8194 (0.31) &	0.1735 (0.07) \\	
&	DPBP &	0.5137 (0.04) &	6.8264 (0.21) &	2.1555 (0.16) &	0.5735 (0.03) &	7.0762 (0.21) &	2.4045 (0.15) &	0.6019 (0.01) &	7.1933 (0.20) &	2.5392 (0.13) \\	
&	PT & 	0.6621 (0.01) &	157.7443 (0.85) &	65.6996 (0.31) &	0.6621 (0.01) &	157.2414 (0.50) &	65.5554 (0.18) &	0.6621 (0.01) &	156.9909 (0.26) &	65.4821 (0.10) \\	
&	W &	0.2982 (0.05) &	3.4876 (0.71) &	0.4979 (0.17) &	0.2759 (0.04) &	3.1490 (0.41) &	0.4145 (0.10) &	0.2599 (0.02) &	2.9623 (0.25) &	0.3631 (0.05) \\	
&	K & 	0.2802 (0.05) &	2.963 (0.92) &	0.4318 (0.23) &	0.2521 (0.04) &	2.4200 (0.69) &	0.3006 (0.14) &	0.2231 (0.03) &	1.9428 (0.45) &	0.2042 (0.08) \\	\hline
S4 &	msBP &	0.2942 (0.04) &	4.3193 (0.82) &	0.5608 (0.12) &	0.2943 (0.03) &	3.6779 (0.57) &	0.5092 (0.05) &	0.2856 (0.02) &	3.4838 (0.35) &	0.4759 (0.04) \\	
&	DPM &	0.3203 (0.06) &	5.0048 (0.80) &	0.7094 (0.25) &	0.3037 (0.05) &	4.4272 (0.63) &	0.5958 (0.17) &	0.2966 (0.04) &	3.9836 (0.59) &	0.5428 (0.15) \\	
&	DPBP &	0.4995 (0.01) &	8.9148 (0.16) &	1.8019 (0.05) &	0.4995 (0.01) &	9.0004 (0.13) &	1.7803 (0.05) &	0.4995 (0.01) &	9.0851 (0.08) &	1.7881 (0.03) \\	
&	PT & 	0.4995 (0.01) &	93.1303 (0.78) &	20.5193 (0.17) &	0.4995 (0.01) &	92.7538 (0.65) &	20.4479 (0.17) &	0.4995 (0.01) &	92.4458 (0.52) &	20.4152 (0.14) \\	
&	W &	0.2990 (0.06) &	5.4053 (0.76) &	0.7075 (0.23) &	0.2831 (0.04) &	4.613 (0.53) &	0.5752 (0.12) &	0.2734 (0.03) &	4.0647 (0.42) &	0.5036 (0.08) \\	
&	K & 	0.3000 (0.04) &	4.3220 (0.83) &	0.5834 (0.14) &	0.2924 (0.03) &	3.799 (0.61) &	0.5143 (0.08) &	0.2834 (0.02) &	3.5222 (0.44) &	0.4744 (0.05) \\	\hline
\end{tabular}
\end{center}
\scriptsize Note: 0.00 stands for ``$<0.01$''
\label{tab:simulation}											
\end{table}											
\end{landscape}

The results of the simulation are reported in Table 1 and Figure \ref{fig:siml1}. 
The proposed method performs better or equally to the best competitor in almost all scenarios and sample sizes.  The worst performance in each case is obtained for mixtures of Polya trees, with overly-spiky density estimates leading to higher distances from the truth.
In Scenario 1 the msBP approach beats all the competitors, except in large sample sizes when single-scale DP mixtures of Bernstein polynomials are comparable. In Scenario 2 the msBP approach is comparable to the frequentist kernel smoother estimator. In scenario 3 the msBP approach is comparable to DP location-scale mixtures and finally, in Scenario 4 our multiscale approach is clearly performing better than any other method.

\begin{figure}
\begin{center}
\includegraphics[page=1, scale=.8]{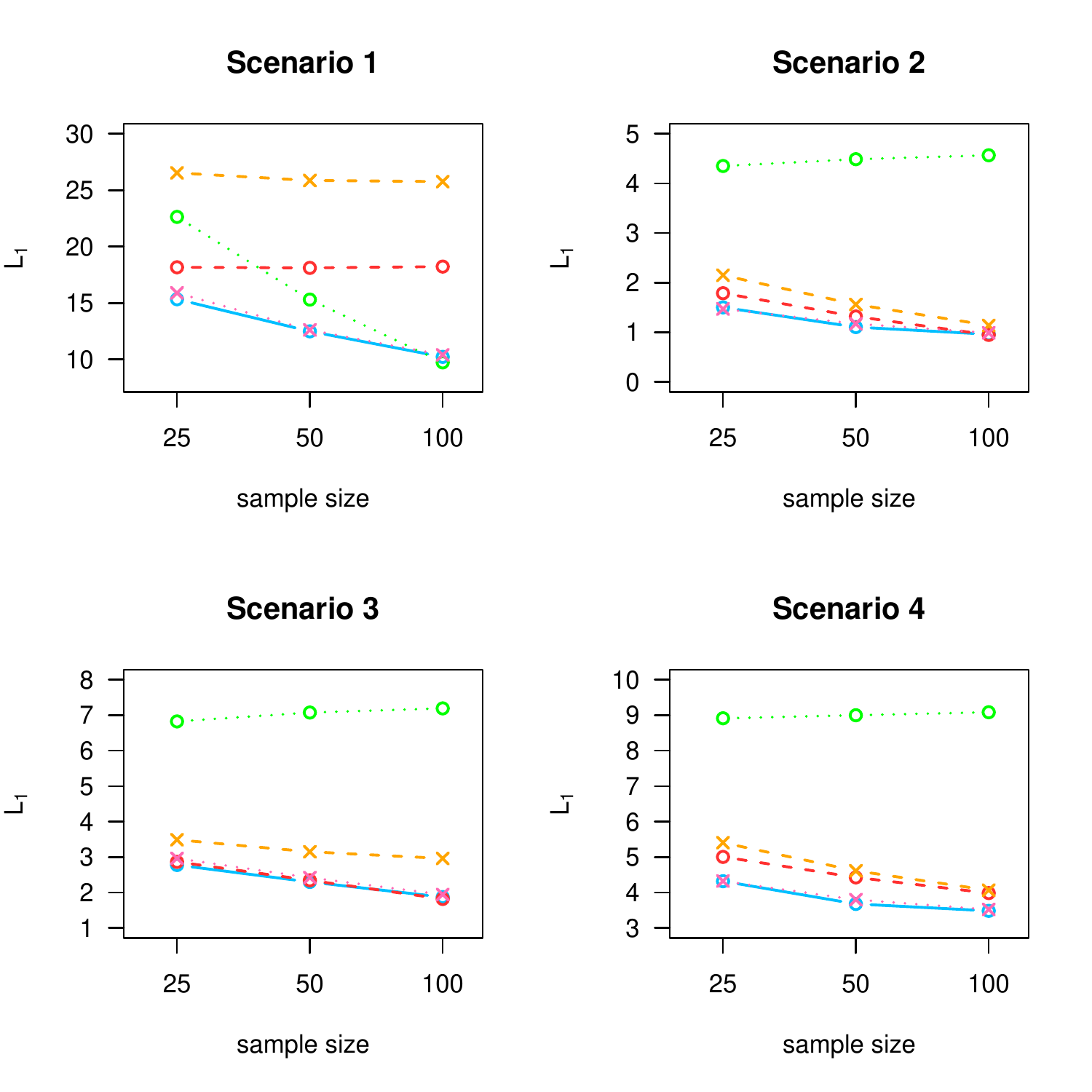}
\end{center}
\caption{Mean $L_1$ distance between the true densities and the posterior msBP estimate (continuous line, circle dots), posterior DP mixture of Gaussians estimate (dashed line, circle dots), posterior DP mixture of Bernstein polynomials estimate (dotted line, circle dots), frequentist wavelet estimate (dashed line, ``$\times$'' dots), and frequentist kernel smoothing estimate (dotted line, ``$\times$'' dots) under the four scenarios. The posterior mixture of Polya trees estimate is far away and it is not reported for graphical reasons.}
\label{fig:siml1}
\end{figure}

\section{Extensions} 

An appealing aspect of the proposed method is ease of generalization to include predictors, hierarchical dependence, time series, spatial structure and so on.  To incorporate additional structure, one can replace model \eqref{eq:mix1} for the stopping and right path probabilities with an appropriate variant.  Similar extensions have been proposed for single resolution mixture models by replacing the beta random variables in a stick-breaking construction with probit regressions \citep{chun:duns:2009, rodr:duns:2011}, logistic regressions \citep{ren:etal:2011} or broader stochastic processes \citep{pati:etal:2013}.  We focus here on one interesting extension to the under-studied problem of Bayesian multiscale inferences on differences between groups.

\subsection{Multiscale testing of group differences}

Motivated by epigenetic data, we propose Bayesian multiscale hypothesis tests of group differences using multiscale Bernstein polynomials.  DNA methylation arrays collect data on epigenetic modifications at a large number of CpG sites.  Let $y_i = (y_{i1},\ldots,y_{ip})^T$ denote the DNA methylation data for patient $i$ at $p$ different sites, with $d_i \in \{0,1\}$ denoting the patient's disease status, either $d_i=0$ for controls or $d_i=1$ for cases.  Current standard analyses rely on independent screening using $t$-tests to assess differences between cases and control at each site.  However, DNA methylation data are constrained to $y_{ij} \in (0,1)$ and tend to have a complex distribution having local spikes and varying smoothness.  

As illustration we focus on nonparametric independent screening; the approach is easily adapted to accommodate dependence across sites.  We center our prior on the uniform as a default. The density of $y_{ij}$ given $d_i=0$ is modeled as in previous sections.   Let $H_{0}: f_{0} = f_{1}$ denote the {\em global} null hypothesis of no difference between groups, with $H_{1}: f_{0} \neq f_{1}$ denoting the alternative. 
Using an msBP representation, $f_0 = f_1$ if the groups share weights over the dictionary of beta densities. If $f_0 \neq f_1$, we may have the same weights on the dictionary elements up to a given scale, so that the densities are equivalent up to that scale but not at finer scales.  With this in mind, let  $H_0^s: f^s_{0} = f^s_{1}$ denote the null hypothesis of no differences between groups at scale $s$, and  $H_{1}^s: f_{0}^s \neq f_{1}^s$ the alternative. As $H_0^0$ is true with probability one, we set $S_{0,1}=0$ and concentrate on $H_0^s$ for $s \ge 1$.


Each of the $n$ subjects in the sample takes a path through the binary tree, stopping at a finite depth.  Let $\mathcal{I}^s = \{ i: s_i \ge s \}$ index the subjects {\em surviving} up to scale $s$ and let $\mathcal{N}^s$ denote the actions of these subjects at scale $s$, including stopping or progressing downward to the left or right for each of the nodes.  Subscripts $(d)$ on $\mathcal{I}^s$ and $\mathcal{N}^s$ denote the restriction to subjects having $d_i=d$.  Conditionally on $H_0^s$, the probabilities for each scale $s$ action are the same in the two groups and the likelihood of actions $\mathcal{N}^s$ is
\begin{align}
\mbox{pr} & (\mathcal{N}^s | H_0^s) 
		 = \int_{\mathcal{T}^{}} \mbox{pr}(\mathcal{N}^{s} | \mathcal{T}^{}) \mbox{pr}(\mathcal{T}^{}|a,b) d \mathcal{T}^{} \notag \\
		& =  \left\{\frac{\Gamma(a+1)}{\Gamma(a)} \frac{\Gamma(2b)}{\Gamma(b)^2} \right\}^{2^{s}} 
		     \int_{\mathcal{T}^{}}  \prod_{h = 1}^{2^{s}} 
				S_{s,h}^{n_{s,h}} (1-S_{s,h})^{\hat{a}_{s,h}-1} 
				R_{s,h}^{\hat{b}_{s,h}-1}  (1 - R_{s,h})^{\hat{c}_{s,h} -1} d \mathcal{T}^{} \notag \\
		& =   \left\{\frac{\Gamma(a+1)\Gamma(2b)}{\Gamma(a) \Gamma(b)^2} \right\}^{2^{s}} 
			 \prod_{h = 1}^{2^{s}} 
			\frac{\Gamma(1 + n_{s,h}) \Gamma(\hat{a}) }{\Gamma(a + v_{s,h} + 1)} 
			\frac{\Gamma(\hat{b}) \Gamma(\hat{c}) }{\Gamma(2b + v_{s,h} - n_{s,h})}, \label{pnest0}
\end{align} 
where $\hat{a}_{s,h} =  a + v_{s,h} - n_{s,h}$, $\hat{b}_{s,h} = b + r_{s,h}$, and $\hat{c}_{s,h} = b + v_{s,h} - n_{s,h} - r_{s,h}$. 
Similarly under $H_1$ we have
\begin{align}
\mbox{pr}(\mathcal{N}^{s} | H_1^{s})  = &\, \mbox{pr}(\mathcal{N}_{(0)}^{s}| H_1^{s})
 \times \mbox{pr}(\mathcal{N}_{(1)}^{s}| H_1^{s}) \notag\\
	= &  \left\{\frac{\Gamma(a+1)\Gamma(2b)}{\Gamma(a) \Gamma(b)^2} \right\}^{2^{2s}} 
			\prod_{h = 1}^{2^{s}} 
			\frac{\Gamma(1 + n_{s,h}^{(0)}) \Gamma(\hat{a}^{(0)}) }{\Gamma(a + v_{s,h}^{(0)} + 1)} 
			\frac{\Gamma(\hat{b}^{(0)}) \Gamma(\hat{c}^{(0)}) }{\Gamma(2b + v_{s,h}^{(0)} - n_{s,h}^{(0)})}\times \notag\\
	& 
			\prod_{h = 1}^{2^s} 
			\frac{\Gamma(1 + n_{s,h}^{(1)}) \Gamma(\hat{a}^{(1)}) }{\Gamma(a + v_{s,h}^{(1)} + 1)} 
			\frac{\Gamma(\hat{b}^{(1)}) \Gamma(\hat{c}^{(1)}) }{\Gamma(2b + v_{s,h}^{(1)} - n_{s,h}^{(1)})}, \label{pnest1}
\end{align}
where $v_{s,h}^{(d)}$ is the number of subjects passing through node $(s,h)$ in group $d$, $n_{s,h}^{(d)}$ is the number of subjects stopping at node $(s,h)$ in group $d$, and $r_{s,h}^{(d)}$ is the number of subjects that continue to the right after passing through node $(s, h)$ in group $d$, with $d=0,1$.

Combining \eqref{pnest0}--\eqref{pnest1} we can obtain a closed form for the posterior probability of $H_0$ being true at scale $s$, given $\mathcal{N}^s_{(0)}$ and $\mathcal{N}^s_{(1)}$:
\begin{align}
\mbox{pr}(H_0^s|\mathcal{N}^s_{(0)}, \mathcal{N}^s_{(1)}) 
	& = \frac{P_0^s\mbox{pr}(\mathcal{N}^s_{(0)}, \mathcal{N}^s_{(1)} | H_0^s) }{ P_0^s\mbox{pr}(\mathcal{N}^s_{(0)}, \mathcal{N}^s_{(1)} | H_0^s)  + (1-P_0^s)\mbox{pr}(\mathcal{N}^s_{(0)}, \mathcal{N}^s_{(1)} | H_1^s)},\label{eq:postH0} 
\end{align}
where $P_0^s$ is our prior guess for the null being true at scale $s$. The global null will be the cumulative product of the $\mbox{pr}(H_0^s|\mathcal{N}^s_{(0)}, \mathcal{N}^s_{(1)}) $ for each scale. An interesting feature of this formulation is to have a multiscale hypothesis testing setup. Indeed the posterior probability of $H_0$ up to scale $\tilde{s}$ will be $\prod_{s\leq\tilde{s}} \mbox{pr}(H_0^s|\mathcal{N}^s_{(0)}, \mathcal{N}^s_{(1)})$ and hence the hypothesis that two groups have the same distribution may have high posterior probability for coarse scales, but can be rejected for a finer scale.

\subsection{Posterior computation}

The conditional posterior probability for $H_0^s$ in \eqref{eq:postH0} is simple, but not directly useful due to the dependence on the unknown $\mathcal{N}^s$ allocations.  To marginalize out these allocations, we modify Algorithm~\ref{algo:gibbs}.  For node $h$ at scale $s$, let $\pi_{s,h}^{(0)}$ denote the weight under $H_0^s$ and $\pi_{s,h}^{(1,d)}$ for $d=0,1$ denote the group-specific weights under $H_1^s$.  At each iteration, the allocation of subject $i$ of group $d$ will be made according to the tree of weights given by
\begin{equation}
	\pi_{s,h}^{(d)} = P(H_0^s|\mathcal{N}^s_{(0)}, \mathcal{N}^s_{(1)}) \pi_{s,h}^{(0)} + 
		\{1 - P(H_0^s|\mathcal{N}^s_{(0)}, \mathcal{N}^s_{(1)})\} \pi_{s,h}^{(1,d)}.
\label{eq:treetest}
\end{equation}
Given the allocation one can calculate all the quantities in \eqref{pnest0}--\eqref{pnest1} and then update the stopping and descending probabilities under $H_0$ and $H_1$ following \eqref{eq:postSR} and the posterior of the null following \eqref{eq:postH0} up to a desired upper scale. 


\section{Application} 
We illustrate our approach on a methylation array dataset for $n$ = 597 breast cancer samples registered at $p = 21{,}986$ CpG sites \citep{tcga:2012}. We test for differences between tumors that are identified as basal-like ($n_0$ = 112) against those that are not ($n_1$ = 485) at each CpG site.  This same problem was considered in a single scale manner by \citet{lock:duns:2014} using finite mixtures of truncated Gaussians.

We run the Gibbs sampler reported in Algorithm~\ref{algo:test} in the Appendix, assuming a uniform prior for $P_0^s$ for each scale $s$. We fixed the maximum scale to 4 as an upper bound, as finer scale tests were not thought to be interpretable.  The sampler is run for 2{,}000 iterations after 1{,}000 burn-in iterations. The chains mix well and converge quickly for all sites and all scales. 

The posterior distribution of $1-P_0^s$ for each scale provides a summary of the overall proportion of CpG sites for which there was a difference between the two groups.  The estimated posterior means for these probabilities were 0.04, 0.07, 0.05 and 0.03, respectively, for scales $1,\ldots,4$.  This suggests that DNA methylation levels were different for a small minority of the CpG sites, which is as expected.  
Examining the posterior probabilites of $H_1^s$ across the 21,986 CpG sites, consistently with the estimates for $1-P_0^s$, we find that scale-specific estimated posterior probabilities are close to zero for most sites.  
Focusing on the 1{,}696 sites for which the overall posterior probability of $H_1$ is greater than $0.5$, we calculated the minimal scale showing evidence of a difference, $\min\{ s: \hat{Pr}(H_1^s|-)>0.5 \}$, with $\hat{Pr}(H_1^s|-)$ denoting the estimated posterior probability.  The proportions of sites having minimal scale equal to $1,2,3,4$ were $47\%, 43\%, 7\%, 3\%$ respectively. 

Figure~\ref{fig:fourgroups} shows $\hat{Pr}(H_1^s|-)$ for these 1{,}696 sites. In the top right quadrant we report those sites having minimal scale equal to $1$. Two different patterns are evident: (1) consistently high 
$\hat{Pr}(H_1^s|-)$, with differences evident at the coarse scale. Site \emph{cg00117172} is among those and its sample distribution is reported in panel (a) of Figure~\ref{fig:examples}. (2) moderate $\hat{Pr}(H_1^s|-)$ for $s=1$, with clear evidence at $s=2$.  Averages of the sites in these two groups are shown with thick dashed lines.

The top right panel, representing sites having minimal scale equal to $2$, presents two patterns: (1) no differences at scale one but clear evidence of $H_1$ at scale two. Site \emph{cg00186954} in panel (b) of Figure~\ref{fig:examples} has this behavior. (2) moderately growing evidence for $H_1$ for increasing scale level. The bottom two panels show results for sites having minimal scale equal to 3 and 4, showing again two different patterns: (1) A group with mild or no evidence for $H_1$ up to scale 3 and 4, respectively (e.g. site \emph{cg20603888} reported in panel (c) of Figure~\ref{fig:examples}), and (2) another group with increasing evidence for increasing scale.  These scale-specific significant tests are interesting in that coarser scale differences are more likely to be biologically significant, while very fine scale differences may represent local changes with minor impact.

\begin{figure}[h]
\begin{center}
\includegraphics[scale=.95]{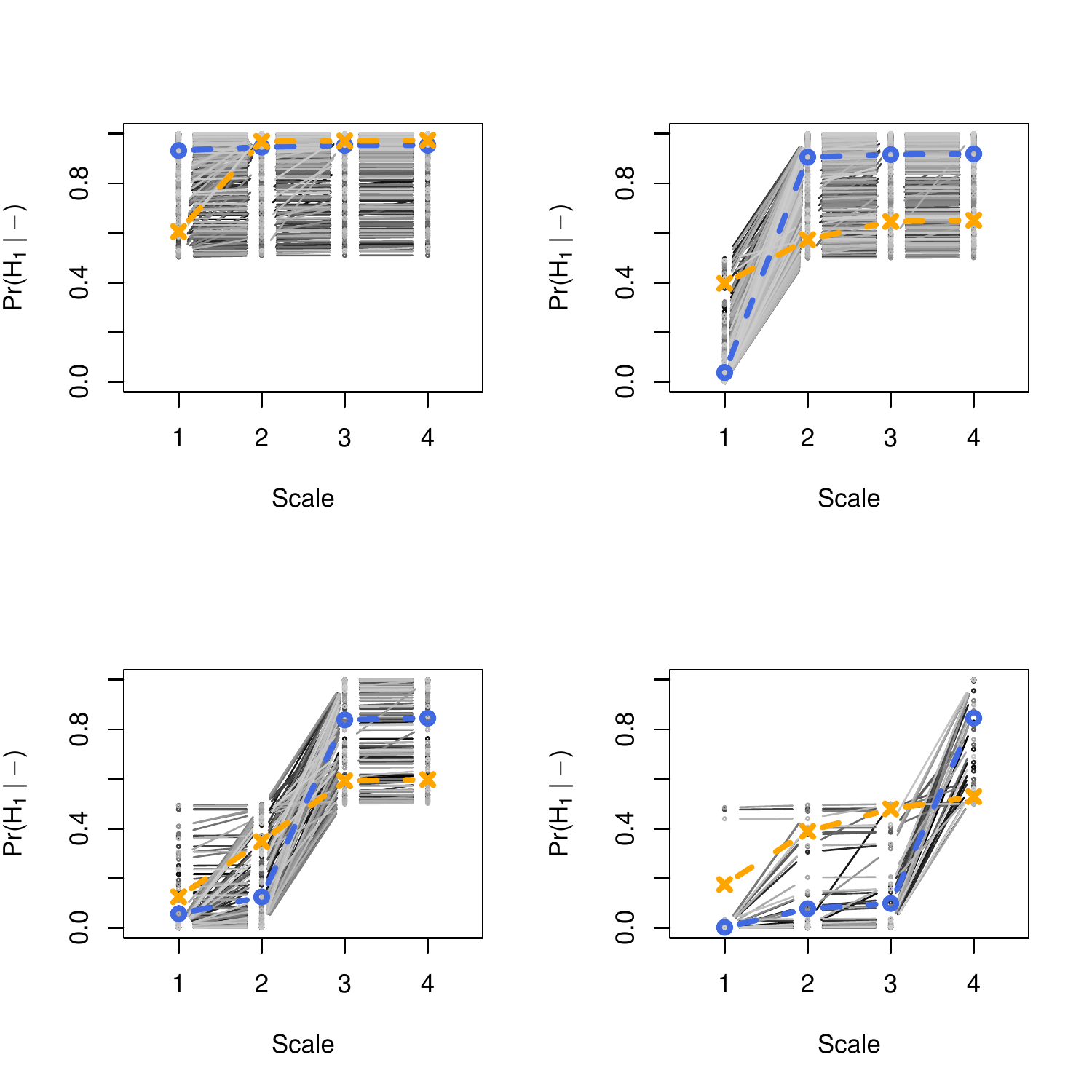}
\end{center}
\caption{Posterior mean probabilities of $H_1$ depending on scale for the 1{,}696 sites, with some evidence of differences in the two groups, grouped in subplots by minimal scale showing $\hat{Pr}(H_1^s|-)>0.5$ for $s=1, \dots 4.$ Within each panel, the thick dashed lines represents the average between the sites in two clusters showing different patterns.}
\label{fig:fourgroups}
\end{figure}
  
\begin{figure}[h]
\begin{center}
\subfigure[]{\includegraphics[scale=.6, page=1]{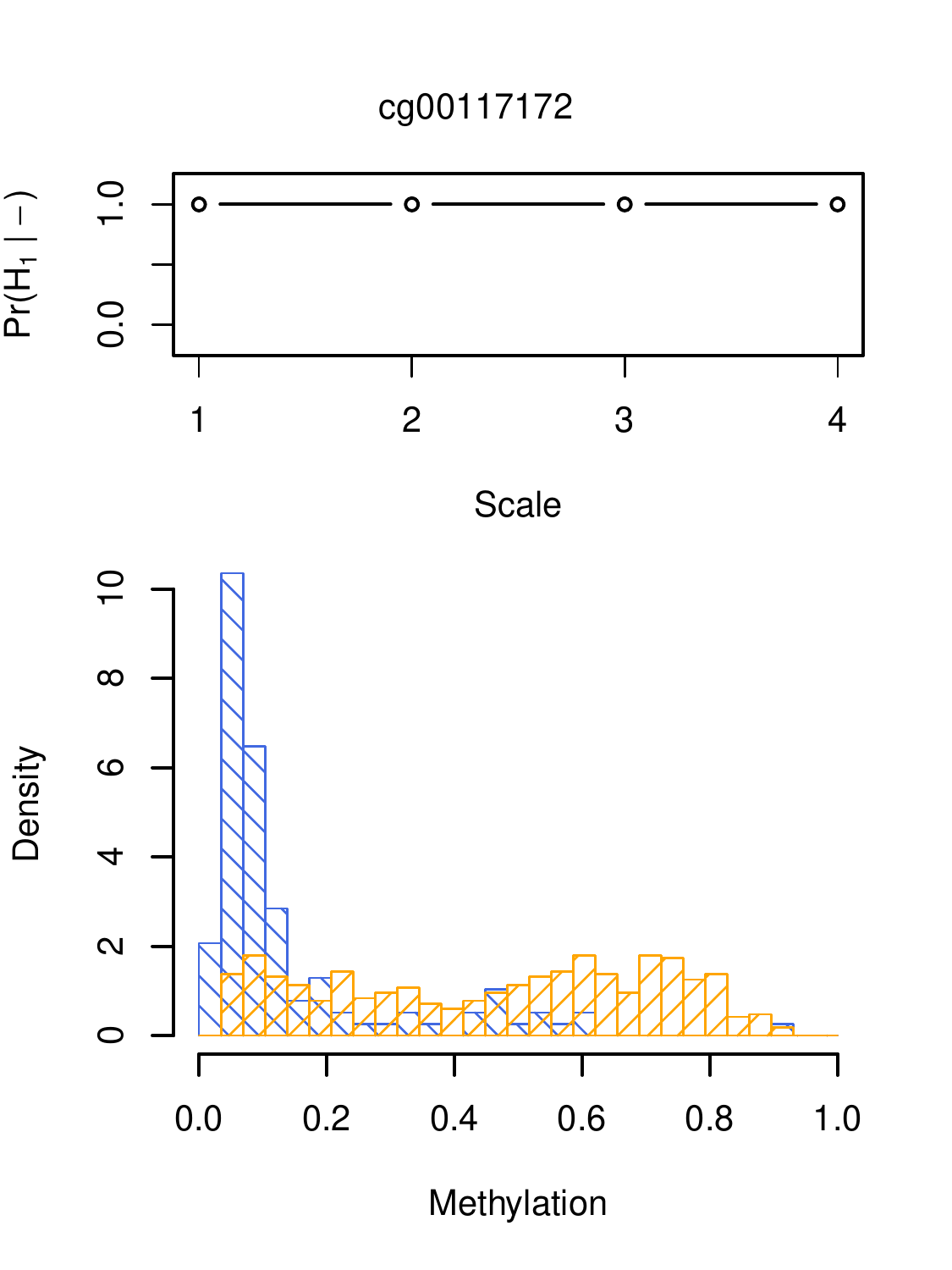}}
\subfigure[]{\includegraphics[scale=.6, page=2]{other4}}
\subfigure[]{\includegraphics[scale=.6, page=1]{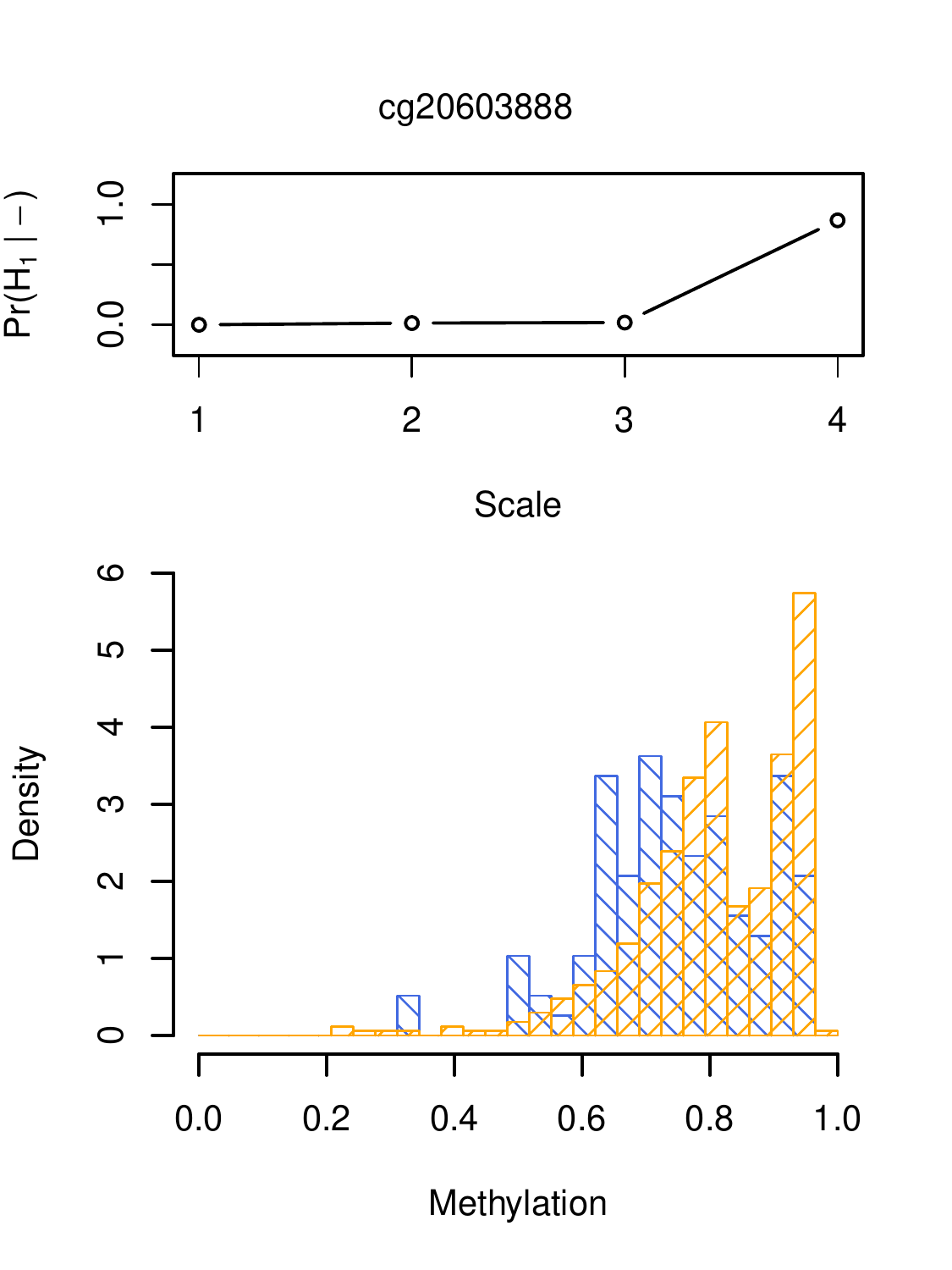}}
\subfigure[]{\includegraphics[scale=.6, page=2]{differenceslock}}
\end{center}
\caption{Histogram of the methylation for the basal (decreasing 45 degree angle shading) and non-basal (increasing 45 degree angle shading) samples for four CpG sites and posterior mean probabilities of $H_1$ in function of scale.}
\label{fig:examples}
\end{figure}



\clearpage
\section{Discussion}

Existing Bayesian nonparametric multiscale tools for density estimation have unappealing characteristics, such as favoring overly­-spiky densities. Our framework overcomes such limitations.  We have demonstrated some practically appealing properties, including simplicity of formulation and ease of computation, and proposed an extension for Bayesian multiscale hypothesis testing of group differences. Multiscale hypothesis testing is of considerable interest in itself, and provides a new view on the topic of nonparametric testing of group differences, with many interesting facets.  For example, it can be argued that in large samples there will always be small local differences in the distributions between groups, which may not be scientifically relevant.  By allowing scale-specific tests, we accommodate the possibility of focusing inference on the range of relevant scales in an application, providing additional insight into the nature of the differences.   We also accommodate scale-specific adaptive borrowing of information across groups in density estimation; extensions to include covariates and hierarchical structure are straightforward.

\section*{Acknowledgement}
The authors thanks Eric Lock for helpful comments on Section 5 and Roberto Vigo for comments on the code implementation. 

\section*{Appendix}

\begin{proof}[Proof of Lemma 1]

For finite $N$ define $\Delta_N = 1 - \sum_{s=0}^N \sum_{h=1}^{2^s} \pi_{s,h}$, for which the following inequality holds:
\begin{equation}
	\Delta_N = \sum_{h=1}^{2^N} \prod_{r\leq N} (1-S_{r,g_{Nhr}})T_{r-1,g_{Nhr}} \leq 2^N \max_{h=1, \dots, 2^N} \prod_{r\leq N} (1-S_{r,g_{Nhr}})T_{r-1,g_{Nhr}}.
\label{eq:delta}
\end{equation}
To establish \eqref{eq:sumtoone}, it is sufficient to take the limit of $\Delta_N$ for $N\to\infty$ and show that it converges to 0 a.s. To this end, take the logarithm of the right hand side of \eqref{eq:delta}, 
\begin{equation}
	\log(\Delta_N) \leq \max_{h=1, \dots, 2^N} \sum_{r\leq N}  \log \left\{ 2^N (1-S_{r,g_{Nhr}})T_{r-1,g_{Nhr}} \right\},
\label{eq:logdelta}
\end{equation}
and notice that for each $h = 1, \dots, 2^N$ we have
\begin{equation}
	E \left\{ 2^N (1-S_{r,g_{Nhr}})T_{r-1,g_{Nhr}} \right\} = 	2^N \left(\frac{a}{a+1}\right) \frac{1}{2^N}  = \frac{a}{a+1}.
\label{eq:explogdelta}
\end{equation}
Therefore taking $N \to \infty$, by Kolmogorov's three series theorem and Jensen's inequality, the argument of the maximum of \eqref{eq:logdelta}, converges to $-\infty$ a.s. for each $h$. Thus $\Delta_N$ converges to 0 a.s. which concludes the proof.
\end{proof}

\begin{proof}[Detail on moments of $F(A)$.]
The expectation of $F(A)$ is simply
\begin{align*}
	E[F(A)] & = E\left[\sum_{s=0}^\infty \sum_{h=1}^{2^s} \pi_{s,h} \int_A \mbox{Be}(y; h, 2^s - h +1) \right]\\
		  & = \sum_{s=0}^\infty \frac{1}{1+a} \left(\frac{a}{1+a} \right)^{s} \frac{1}{2^{s}}  \sum_{h=1}^{2^s} \int_A \mbox{Be}(y; h, 2^s - h +1) \\
		  & = \sum_{s=0}^\infty  \frac{1}{1+a} \left(\frac{a}{1+a} \right)^{s}  \lambda(A)   \\
		  & = \lambda(A) \sum_{s=0}^\infty  \frac{1}{1+a} \left(\frac{a}{1+a} \right)^{s} = \lambda(A),
\end{align*}
where the third equality follows from the fact that the average measure over scale $s$ beta dictionary densities of any region $A$ equals the Lebesgue measure of $A$. 
\end{proof}

\begin{proof}[Proof of Lemma \ref{lem:tvd}]
First note that twice the total variation distance between two measures $P^s$ and $P$ equals the $L_1$ distance between the densities $f^s$ and $f$. For the expectation, the following holds
\[
  E \bigg[ \int \bigg| f^s( y) - f(y)  \bigg| dy \bigg]  =   \int E \bigg[ \bigg| f^s( y) - f(y)  \bigg| \bigg] dy  
\]
by Fubini's theorem. Now since 
\[
   \bigg| f^s( y) - f(y)  \bigg|  =   f^s( y) - f(y)  + 2 \max\{  f( y) - f^s(y), 0 \},
\]
it is sufficient to prove that the expectation of $f^s( y) - f(y)$ is null. This can be done, noting that for each $y \in [0,1]$ and for each scale $s$, the quantity $2^{-s} \sum_{h=1}^{2^s} \mbox{Be}( y; h, 2^s - h + 1) = 1$. Hence
\begin{align*}
\sum_{h=1}^{2^s} & E[\tilde{\pi}_{s,h}]   \mbox{Be}( y; h, 2^s - h + 1) - 
\sum_{l=s}^{\infty}\sum_{h=1}^{2^l} E[\pi_{l,h}] \mbox{Be}( y; h, 2^l-h+1)  = \\
&  = \sum_{h=1}^{2^s} E[\tilde{\pi}_{s,h} - \pi_{s,h}]   \mbox{Be}( y; h, 2^s - h + 1) - 
\sum_{l=s+1}^{\infty}\sum_{h=1}^{2^l} E[\pi_{l,h}] \mbox{Be}( y; h, 2^l-h+1)   \\
& = \left(\frac{a}{1+a}\right)^{s+1} \frac{1}{2^s}\sum_{h=1}^{2^s}   \mbox{Be}( y; h, 2^s - h + 1) - 
\sum_{l=s+1}^{\infty} \frac{1}{1+a} \left(\frac{a}{1+a}\right)^{l} \frac{1}{2^l} \sum_{h=1}^{2^l}  \mbox{Be}( y; h, 2^l-h+1)   \\
& = \left(\frac{a}{1+a}\right)^{s+1} - 
\sum_{l=s+1}^{\infty} \frac{1}{1+a} \left(\frac{a}{1+a}\right)^{l}    
= \left(\frac{a}{1+a}\right)^{s+1} - \left(\frac{a}{1+a}\right)^{s+1}  = 0,
\end{align*}
which concludes the first part of proof. 
Now consider
\begin{align*}
& \int \left| 
  \sum_{l=0}^s       \sum_{h=1}^{2^l} \tilde{\pi}_{l,h} \mbox{Be}( y; h, 2^s - h + 1) - 
  \sum_{l=0}^{\infty}\sum_{h=1}^{2^l} \pi_{l,h}         \mbox{Be}( y; h, 2^l-h+1)  \right| d y \\
=& \int \left| 
  \sum_{l=0}^{\infty}\sum_{h=1}^{2^l} \left(\tilde{\pi}_{l,h} - \pi_{l,h}\right)   \mbox{Be}( y; h, 2^s - h + 1)  \right| d y \\
\leq& \int 
  \sum_{l=0}^{\infty}\sum_{h=1}^{2^l} \left|  \left(\tilde{\pi}_{l,h} - \pi_{l,h}\right)   \mbox{Be}( y; h, 2^s - h + 1)  \right| d y \\
=&  
  \sum_{l=0}^{\infty}\sum_{h=1}^{2^l} \left|  \left(\tilde{\pi}_{l,h} - \pi_{l,h}\right)\right|  \int  \mbox{Be}( y; h, 2^s - h + 1)  d y 
=  
  \sum_{l=0}^{\infty}\sum_{h=1}^{2^l} \left|  \left(\tilde{\pi}_{l,h} - \pi_{l,h}\right)\right|, 
\end{align*}
where the inequality holds since for each $y$ the absolute values of the sum is less than the sum of the absolute values. Since the first moment is null the variance is
\begin{align*}
 E \left[ \left\{ \int \left| f^s( y) - f(y)  \right| dy \right\}^2 \right] & = 
 E \left[ \left(   \sum_{l=0}^{\infty}\sum_{h=1}^{2^l} \left|\tilde{\pi}_{l,h} - \pi_{l,h}\right|\right)^2 \right] \\ & =
 E \left[ \left(   \sum_{h=1}^{2^s} \left|  \tilde{\pi}_{s,h} - \pi_{s,h}\right| + \sum_{l=s+1}^{\infty}\sum_{h=1}^{2^l} \pi_{s,h}  \right)^2 \right] \\ & \leq
 2 E \left[ \left(   \sum_{h=1}^{2^s} \left|  \tilde{\pi}_{s,h} - \pi_{s,h}\right|\right)^2 + \left(\sum_{l=s+1}^{\infty}\sum_{h=1}^{2^l} \pi_{s,h}  \right)^2 \right].
\end{align*}
We study separately the expecations of the two summands above. For each $h=1\dots,2^s$, $ \tilde{\pi}_{s,h} \geq \pi_{s,h}$, thus the fist expectation is
\begin{align*}
 E \left\{ \left(   \sum_{h=1}^{2^s} \tilde{\pi}_{s,h} - \pi_{s,h}\right)^2 \right\} & \leq  E \left\{ \left(   \sum_{h=1}^{2^s} \tilde{\pi}_{s,h}\right)^2 + \left(   \sum_{h=1}^{2^s}\pi_{s,h}\right)^2 \right\} \\
& \leq E  \left(   \sum_{h=1}^{2^s} \tilde{\pi}_{s,h} +   \sum_{h=1}^{2^s}\pi_{s,h}\right)\\
& = \left( \frac{a}{1+a}\right)^s +  \frac{1}{1+a} \left( \frac{a}{1+a}\right)^s,
\end{align*}
where the first inequality holds removing twice the cross product, and the second since the quantities are strictly less than one.
The second expectation is simply
\begin{align*}
E\left\{\left(\sum_{l=s+1}^{\infty}\sum_{h=1}^{2^l} \pi_{s,h}  \right)^2\right\} 
\leq E\left(\sum_{l=s+1}^{\infty}\sum_{h=1}^{2^l} \pi_{s,h}  \right) = \left( \frac{a}{1+a}\right)^{s+1}.
\end{align*}
It follows that the variance is less than $2\{a/(1+a)\}^s$, that concludes the proof.
\end{proof}

\begin{algorithm}
\caption{Gibbs sampler steps for posterior computation for multiscale hypothesis testing of group differences using msBP prior}
\begin{algorithmic}
\footnotesize
\FOR{ $j = 1, \dots, p$}
\STATE Compute the threes for the node allocation according to \eqref{eq:treetest}. 
\FOR{ $i = 1, \dots, n$}
\STATE assign observation $i$ at site $j$ to a cluster $(s_i, h_i)$ as in Algorithm~\ref{algo:postcluster} using the tree of weights of last step
\ENDFOR
\STATE compute $n_{s,h}$, $v_{s,h}$, and $r_{s,h}$;
\STATE compute $n^{(j)}_{s,h}$, $v^{(j)}_{s,h}$, and $r^{(j)}_{s,h}$ for $j=0,1$;
\STATE let $s_{\text{MAX}}$ the maximum occupied scale;
\FOR{ $s = 0, \dots, s_{\text{MAX}}$}
\FOR{ $h = 1, \dots, 2^s$}
\STATE update $S_{s,h} \sim \Be{1+n_{s,h}}{a + v_{s,h} - n_{s,h}}$,  $R_{s,h} \sim \Be{b+r_{s,h}}{b + v_{s,h} - n_{s,h} - r_{s,h} }$
\STATE update $S^{(0)}_{s,h} \sim \Be{1+n^{(0)}_{s,h}}{a + v^{(0)}_{s,h} - n^{(0)}_{s,h}}$,  $R^{(0)}_{s,h} \sim \Be{b+r^{(0)}_{s,h}}{b + v^{(0)}_{s,h} - n^{(0)}_{s,h} - r^{(0)}_{s,h} }$
\STATE update $S^{(1)}_{s,h} \sim \Be{1+n^{(1)}_{s,h}}{a + v^{(1)}_{s,h} - n^{(1)}_{s,h}}$,  $R^{(1)}_{s,h} \sim \Be{b+r^{(1)}_{s,h}}{b + v^{(1)}_{s,h} - n^{(1)}_{s,h} - r^{(1)}_{s,h} }$
\ENDFOR
\ENDFOR
\STATE compute the trees of weights under $H_0$ and $H_1$ for the two groups
\FOR{ $s = 0, \dots, s_{\text{MAX}}$}
\STATE compute $P_m^s = \mbox{pr}(H_0^s| \mathcal{N}^s_{(0)}, \mathcal{N}^s_{(1)})$ as in \eqref{eq:postH0}.
\ENDFOR
\ENDFOR
\STATE Draw $P_0^s  \sim \mbox{Be}(1 + \sum_{m=1}^M P_m^s, 1+M- \sum_{m=1}^M P_m^s)$
\end{algorithmic}
\label{algo:test}
\end{algorithm}


\bibliographystyle{apalike}
\bibliography{references}
\end{document}